\documentclass[a4paper,11pt]{article}
\usepackage[dvips]{graphicx}
\usepackage{amssymb}
\usepackage{amsmath}

\usepackage{cite}

\usepackage{sgame}

\usepackage{hyperref}

\usepackage{amsmath}
\usepackage{amscd}
\usepackage{amsthm}
\usepackage{amssymb}
\usepackage{latexsym}
\usepackage{eufrak}
\usepackage{euscript}
\usepackage{epsfig}
\usepackage{graphics}
\usepackage{array}
\usepackage{enumerate}

 \addtocounter{equation}{0}

\newcommand{\be}{\begin{equation}}
\newcommand{\bel}[1]{\begin{equation}\label{#1}}
\newcommand{\qe}{\end{equation}}
\newcommand{\ee}{\end{equation}}
\newcommand{\eeq}{\end{equation}}
\newcommand{\ba}{\begin{eqnarray}}
\newcommand{\ea}{\end{eqnarray}}

\theoremstyle{theorem}
\newtheorem{theo}{Theorem}
\theoremstyle{theorem}
\newtheorem{algo}{Algorithm}
\theoremstyle{theorem}

\theoremstyle{remark}

\theoremstyle{remark}

\theoremstyle{corollary}

\theoremstyle{lemma}
\theoremstyle{definition}
\newtheorem{defi}{Definition}

\usepackage{color}

\begin{document}

\title{Periodic Strategies.  \\
A  New Solution Concept and an Algorithm for Non-Trivial Strategic
Form Games}
\author{V.K. Oikonomou$^{1}$\thanks{
voiko@physics.auth.gr, v.k.oikonomou1979@gmail.com}, J.
Jost$^{1,2}$\thanks{
jost@mis.mpg.de}\\
$^{1}$Max Planck Institute for Mathematics in the Sciences\\
Inselstrasse 22, 04103 Leipzig, Germany\\
 $^{2}$Santa Fe
Institute, New Mexico, USA }\maketitle
\begin{abstract}

We introduce a new solution concept, called periodicity,  for
selecting optimal strategies in strategic form games. This
periodicity solution concept yields new insight into  non-trivial
 games. In
mixed strategy strategic form games, periodic solutions yield values
for the utility function of each player that are equal to the Nash
equilibrium ones. In contrast to the Nash strategies, here the
payoffs of each player are  robust against what the opponent plays.
Sometimes, periodicity strategies yield higher utilities, and
sometimes the Nash strategies do, but often the utilities of these
two strategies coincide. We formally define and study periodic
strategies in two player perfect information strategic form games
with pure strategies and we prove that every non-trivial finite game
has at least one periodic strategy, with non-trivial meaning
non-degenerate payoffs.  In some classes of games   where mixed
strategies are used, we identify quantitative features. Particularly
interesting are the implications for  collective action games, since
there  the collective action strategy can be incorporated in a
purely non-cooperative context. Moreover, we address the periodicity
issue when the players have a continuum set of strategies available.
\end{abstract}

\section{Motivation for Periodicity and Periodic Strategies--A non-cooperative Concept}

We introduce a new concept in game theory that  is an inherent
characteristic of every non-trivial, finite action-player,
simultaneous, strategic form game with or without perfect
information. Here,  non-trivial means  non-degenerate payoffs for
the players. Should a game have degenerate pay-offs, it can be
perturbed to one with non-degenerate payoffs.

We shall call this new mathematical concept ``periodicity'' and we
shall describe  both the mathematical implications and its
applications  in specific games.

For this  periodicity concept,  non-cooperativity plays an essential
role. The key aspect is that each player tries to \textit{maximize
his own payoff}, by observing and predicting which action of his
opponent will make his payoff maximized. The situation is
non-cooperative as each player tries to maximize his own payoff, but
there is an important difference to standard game theory. For
finding periodic strategies, each player  ''scans'' his opponent's
actions, builds hierarchical belief systems on these actions, by
assigning corresponding probabilities and investigates which of his
opponent's strategies will maximize his own payoffs. While this is
different from standard game theory, we shall see that it is
self-consistent. Moreover, it will result in payoffs that are at
least as high as those reached by Nash type strategies.
\medskip

Here is an example, which we will analyze in detail in a later
section. Consider the game named ''Test Game'' appearing in the
table below, which is a two player strategic form game, played
simultaneously.
\begin{table}[h]\notag
\centering
  \begin{tabular}{| l |l |l | }
    \hline
       &  $b_1$ & $b_2$ \\ \hline
  $a_1$ & 2,5 & 50,6   \\ \hline
    $a_2$ & 3,10 & 2,5\\
    \hline
 \end{tabular}
\caption{Test Game} \notag
\end{table}
Player A then can use mixed strategies of the form
\begin{equation}\label{mixain1}
x_{\sigma}=pa_1+(1-p)a_2
\end{equation}
and correspondingly for  B:
\begin{equation}\label{mixbin2}
y_{\sigma}=qb_1+(1-q)b_2.
\end{equation}
The utilities of the players are
\begin{equation}\label{uti}
{\mathcal{U}}_A(p,q)=2pq +50p(1-q) +3(1-p)q + 2(1-p)(1-q),
{\mathcal{U}}_B(p,q)=5pq +6p(1-q) +10(1-p)q + 5(1-p)(1-q).
\end{equation}
This game has two pure and one mixed Nash equilibrium; the latter is
given by $(p_N^*=\frac{5}{6},q_N^*=\frac{48}{49})$. When B plays the
Nash value $q_N^*$, A's utility
\begin{equation}\label{mixjfgj1b}
{\mathcal{U}}_A(p,q_N^*=48/49)=\frac{146}{49}
\end{equation}
is independent of his own strategy $p$, and likewise, when A plays
$p_N^*$, B's utility
\begin{equation}\label{mikjhx12}
{\mathcal{U}}_B(p_N^*=5/6,q)=\frac{35}{6}
\end{equation}
is independent of $q$. We now observe that there are values
$(p_p^*=1/49,q_p^*=1/6)$ (with the subscript $p$ standing for
``periodicity''), with the property that when $A$ plays $p_p^*$,
then his utility
\begin{equation}\label{mixjfgj1a}
{\mathcal{U}}_A(p_p^*=1/49,q)=\frac{146}{49}
\end{equation}
now is independent of the opponent's value $q$. Similarly, B's
${\mathcal{U}}_B(p,q_p^*=1/6)=\frac{35}{6}$ is independent of what A
does. We note that the utilities \eqref{mixjfgj1b} and
\eqref{mixjfgj1a} agree, and the same holds for the utilities of B.

Two insights emerge from this example. First, players can seek
equilibria or stationary values of  their utility consistently
w.r.t. either variable, the own action or that of the opponent. In
one case, we obtain the Nash equilibrium, where it does not matter
what one plays as long as the opponent sticks to his action. In the
other case, we obtain an equilibrium where it does not matter what
the opponent does as long as a player sticks to his own action. The
latter seems better than the former. In terms of market economics,
this could be of some importance, since each player who adopts some
mixed strategy, by choosing to play such a strategy can earn the
optimal payoff (equal to Nash), without depending on the opponents
actions. As we shall analyze in more detail, this depends on
optimizing not w.r.t. the own action, but w.r.t. that of the
opponent. This may seem strange, like wishful thinking, but as we
shall see, when both players consistently do that, they will do at
least as well as when playing Nash. In fact, typically, as in this
example, the equilibrium utilities are the same.

More generally, we shall see that through such a process, cycles
emerge, as in the rationalizable strategies of Bernheim
\cite{bernheim} and Pearce \cite{pearce}. When trying to  compute
the Nash value, a player optimizes his own strategy for any action
of the opponent. When that the process is iterated, the players will
arrive at a cycle of rationalizable strategies, and that latter
class includes the Nash equilibria.  Conversely, as we shall see,
for arriving at the values $p_p^*,q_p^*$, each player asks for that
strategy of the opponent that is best for his current action. That
is, A optimizes not his own action, but that of his opponent B. When
B then takes that strategy as his starting point and computes that
response of A that is best for him, that is, B when playing that
strategy, we arrive at a new strategy of A, and the process can be
repeated and iterated.

The  fundamental concept in non-cooperative game theory
\cite{tirole,newref1,newref2,ozzy,geanakoplos,neu}, the Nash
equilibrium, is one of the most widely and commonly used solution
concepts that predict the outcome of a strategic interaction in the
social sciences. A pure-strategy Nash equilibrium is an action
profile with the important property that no single player can obtain
a higher payoff by deviating unilaterally from this strategy
profile. Based on the rationality of the players, a Nash strategy is
a steady state of strategic interaction. In a strategic form
2-player game with only finitely many actions, assumed to have
non-degenerate payoffs for simplicity, there may exist both pure and
mixed Nash equilibria. In a pure equilibrium, each player chooses
some definite action which is the best response to the opponent's
action, and conversely. Thus, no player can unilaterally change his
action without decreasing his payoff. In a mixed equilibrium,
instead, each player has a probability distribution for his actions
which again is optimal in view of the opponent's distribution. A
pure equilibrium can be considered as a limit case of a mixed one
where all probabilities are either 0 or 1. Mixed cases with
probabilities strictly between 0 and 1 can be computed with
differential calculus, by maximizing w.r.t. those probabilities,
because the payoffs depend differentiably on them (this is a
standard assumption that we shall also make implicitly).

However, as Bernheim notes in his paper \cite{bernheim}, the Nash
equilibrium is neither a necessary consequence of rationality nor a
reasonable empirical proposition. Despite the valuable contributions
that the Nash equilibrium offers to non-cooperative games, there is
a  refinement, the rationalizability solution concept
\cite{rat1,rat2,rat3,rat4,rat5,rat6,rat7,rat8,rat10,rat11,rat12,rat13,rat14,rat15,rat16,rat17}.
The idea of this concept is the following. An action of a player,
say A,  is rationalizable if it is the best response to some action
of the opponent B. In turn, B's action should be the best response
to some action of A, and so on. When we iterate this in a game with
only finitely many actions, then eventually, such rationalizable
actions will repeat themselves, and we obtain a periodic cycle. Our
periodicity concept looks similar, with the important difference,
that a player no longer computes his own best response to an
opponent's action, but rather seeks that opponent's action that is
best for him, given his current action. Again, the process can be
iterated, and we shall then arrive at periodic cycles.

In strategic form games, rationalizability is based on the fact that
each player views his opponent's choices as uncertain events, each
player complies to Savage's axioms of rationality and this fact is
common knowledge \cite{bernheim}. The rationalizability concept
appeared independently in Bernheim's \cite{bernheim} and Pearce's
work \cite{pearce} (a predecessor of the two papers was Myerson's
work \cite{myerson}). Subsequently, the rationalizability solution
concept has been analyzed and refined in various games, both static
and dynamic. For an important stream of papers see
\cite{rat1,rat2,rat3,rat4,rat5,rat6,rat7,rat8,rat10,rat11,rat12,rat13,rat14,rat15,rat16,rat17}
and references therein.

The Nash equilibrium and its refinements are statements about the
existence of a fixed point in every game. In this paper we shall
present another mathematical property of finite player, finite
actions, simultaneous strategic form games, which we shall call
periodicity. Periodicity is a solution concept with interesting
quantitative implications.

The purpose of this paper is to study periodic strategies and
investigate the consequences of periodicity in various cases of
perfect information strategic form games. The terms ``periodic'' and
``periodicity'' indicate that there exist self-maps $\mathcal{Q}$ of
the players' strategy spaces with $\mathcal{Q}^n=1$ for some $n$
$\in$ $N$. The rationalizable strategies of Bernheim and Pearce are
also periodic in that sense.  In our case, periodic strategies arise
from an optimization scheme that is different from that underlying
the rationalizable ones.

\section{Periodic Strategies in Strategic Form Games--Definitions--Pure Strategies Case }
\subsection{Introduction to the Periodicity Concept}
We restrict our present study to simultaneous, strategic form games
with two players A and B, in the context of perfect information,
assuming that the game is played only once and also that each player
has only finitely many actions available. We start with pure
strategies only, before including mixed actions in a later section.
The strategic form game is then defined by:
\begin{itemize}
\item The strategy spaces of
players $A,B$, denoted by  $\mathcal{M}(A)=\{a_1,a_2,a_3,...,a_N\}$
and $\mathcal{M}(B)=\{b_1,b_2,b_3,...,b_N\}$ (for simplicity, we
shall usually assume that $N=2$, i.e., that each player has only two
choices for his actions), and
\item the payoff functions
$\mathcal{U}_i:\mathcal{M}(A)\times \mathcal{M}(B)\rightarrow
\mathbb{R}$, $i={A,B}$. We shall assume, for simplicity again, that
they are non-degenerate in the sense that different actions yield
different pay-offs.
\end{itemize}
We then have the periodicity algorithm
\begin{itemize}
  \item Start from Player A and his first action $a(0)$. Seek that
    action $b(1)$
  available actions $b_1,b_2,b_3,...,b_N$ of B for which  the corresponding payoff of player A
  $\mathcal{U}_A(a(0),b(1))$ is maximized, that is, larger than the
  payoff $\mathcal{U}_A(a(0),b)$ for any other action $b\neq
  b(1)$. (Recall that we assume non-degeneracy of payoffs, so there is
  a unique such $b(1)$.)

  \item For the action $b(1)$ found in the previous step, now seek
    that action $a(1)$ that maximizes B's payoff, that is,
    $\mathcal{U}_B(a(1),b(1)) > \mathcal{U}_B(a,b(1))$ for any other
    $a\neq a(1)$.

  \item When iteratively actions $a(k),b(k)$ have been chosen, seek
    that $b(k+1)$ for which
    $\mathcal{U}_A(a(k),b(k+1))>\mathcal{U}_A(a(k),b)$ for any $b\neq
    b(k+1)$, and then that $a(k+1)$ for which  $\mathcal{U}_B(a(k+1),b(k+1)) > \mathcal{U}_B(a,b(k+1))$ for any other
    $a\neq a(k+1)$.

  \item Continue until $a(k)=a(\ell)$ or $b(k)=b(\ell)$ for some $\ell
    <k$.
\end{itemize}
Since at each step, by non-degeneracy, the selected action of the
opponent is unique, the procedure will then repeat itself, that is,
become periodic. Hence the name ``periodic solution''.

\subsection{Two Player, Perfect Information Strategic Form Games}

Formalizing the preceding, we define two continuous maps between the
strategy spaces,
\begin{equation}\label{phi1}
\varphi_1:\mathcal{M}(A)\rightarrow
\mathcal{M}(B){\,}{\,}{\,}{\,}{\,}{\,}{\,}\varphi_2:\mathcal{M}(B)\rightarrow
\mathcal{M}(A).
\end{equation}
They are defined in terms of  the payoff functions by the following
inequalities
\begin{align}\label{payoffreq}
&\mathcal{U}_A(x,\varphi_1(x))> \mathcal{U}_A(x,y_1)
{\,}{\,}{\,}{\,}{\,}{\,}& \forall {\,}y_1{\,}\in
{\,}\mathcal{M}(B)\backslash\{\varphi_1(x)\}\\ \notag
&\mathcal{U}_B(\varphi_2(y),y)> \mathcal{U}_B(x_1,y) {\,}& \forall
{\,}x_1{\,}\in {\,}\mathcal{M}(A)\backslash\{\varphi_2 (y)\}.
\end{align}
We can achieve the strict inequalities here because we assume that
the pay-off tables are nondegenerate. Iteratively, we obtain for any
positive integer $k$
\begin{align}\label{payoffreq1}
\notag &\mathcal{U}_B((\varphi_2 \circ
\varphi_1)^k(x),\varphi_1\circ (\varphi_2 \circ
\varphi_1)^{k-1}(x))> \mathcal{U}_B(x_1,\varphi_1\circ (\varphi_2
\circ \varphi_1)^{k-1}(x))
\\ \notag & \forall {\,}{\,}x_1{\,} \in {\,}\mathcal{M}(A)\backslash\{(\varphi_2 \circ \varphi_1)^k(x)\}
\\
&\mathcal{U}_A((\varphi_2 \circ \varphi_1)^k(x),\varphi_1\circ
(\varphi_2 \circ \varphi_1)^{k}(x))> \mathcal{U}_A((\varphi_2 \circ
\varphi_1)^kx,y_1)
\\ \notag & \forall {\,}{\,}y_1{\,}\in {\,}\mathcal{M}(B)\backslash\{\varphi_1\circ (\varphi_2 \circ \varphi_1)^{k}(x)\}.
\end{align}
Thus, in each step, a player seeks the opponent's action that
maximizes his own pay-off given his current action. When there are
only finitely many strategies available, as we are currently
assuming, then necessarily the strategies that occur in this chain
will repeat themselves after finitely many steps. That is, after
finitely many steps, the players turn into a periodic cycle that can
be represented  by the diagram
\begin{equation}\label{periodicitysimplex}
x{\,}{\,}{\,}\xrightarrow{P}{\,}{\,}{\,}
\varphi_1(x){\,}{\,}{\,}\xrightarrow{P}{\,}{\,}{\,} \varphi_2\circ
\varphi_1 (x){\,}{\,}{\,}\xrightarrow{P} \cdots
\xrightarrow{P}{\,}{\,}{\,} x.
\end{equation}
 We shall call such an  action $x$ that repeats itself after finitely many steps periodic. The minimal number $n(x)$ of such steps is called the periodicity number of $x$. We denote the set of periodic actions of
player $A$ by $\mathcal{P}(A)$ and those of  player B by
$\mathcal{P}(B)$. For periodic actions, the operator
\begin{equation}\label{oper1}
\mathcal{Q}=\varphi_2 \circ \varphi_1
\end{equation}
satisfies  $\mathcal{Q}^nx=x$. In terms of the operator
$\mathcal{Q}$, the  inequalities  (\ref{payoffreq1}) become
\begin{equation}\label{lastineq}
\mathcal{U}_A({\mathcal{Q}}^k(x),\varphi_1\circ
{\mathcal{Q}}^{k}(x))> \mathcal{U}_A({\mathcal{Q}}^{k}(x),y_1)
\forall {\,}{\,}y_1{\,}\in
{\,}\mathcal{M}(B)\backslash\{\varphi_1\circ (\varphi_2 \circ
\varphi_1)^{k}(x)\}
\end{equation}
Likewise, for  player B, we have the operator
\begin{equation}\label{oper2}
\mathcal{Q}'=\varphi_1 \circ \varphi_2
\end{equation}
In this case,  the inequalities (\ref{payoffreq1}) take the form:
\begin{align}\label{paytheplayer}
&\mathcal{U}_A((\varphi_1 \circ \varphi_2)^k(y),\varphi_2\circ
(\varphi_1 \circ \varphi_2)^{k-1}(y))>
\mathcal{U}_A(y_1,\varphi_2\circ (\varphi_1 \circ
\varphi_2)^{k-1}(y)) \\ \notag & \forall {\,}{\,}y_1{\,}\in
{\,}\mathcal{M}(B)\backslash\{\varphi_1 \circ \varphi_2)^k(y)\}
\\ \notag
&\mathcal{U}_B((\varphi_1 \circ \varphi_2)^k(y),\varphi_2\circ
(\varphi_1 \circ \varphi_2)^{k}(y))> \mathcal{U}_B((\varphi_1 \circ
\varphi_2)^{k}(y),x_1) \\ \notag &\forall {\,}{\,}x_1{\,}\in {\,}
\mathcal{M}(A)\backslash\{\varphi_2\circ (\varphi_1 \circ
\varphi_2)^{k}(y)\}
\end{align}
The last inequality can be written in terms of the operator
$\mathcal{Q}'$, as
\begin{equation}\label{lastineq}
\mathcal{U}_B({\mathcal{Q}'}^k(y),\varphi_2\circ
{\mathcal{Q}'}^{k}(y))>
\mathcal{U}_B({\mathcal{Q}'}^{k}(y),x_1)\quad \forall
{\,}{\,}x_1{\,}\in {\,}\mathcal{M}(A)\backslash\{\varphi_2\circ
{\mathcal{Q}'}^{k}(y)\}
\end{equation}
While the  periodic strategies in our above sense and the
rationalizable strategies are derived from different optimization
schemes, those  rationalizable strategies that are also periodic are
particularly interesting. We will analyze  some characteristic
examples  at the end of this section. Note that the procedure
described by relations (\ref{payoffreq}) and (\ref{paytheplayer})
does not stipulate that the maps $\varphi_1$ and $\varphi_2$ are
best responses to some action.
 Take for example the first of the inequalities
(\ref{payoffreq}). It means  that, by assuming that player A plays
$x$, we seek in the action set of player B for an action
$\varphi_1(x)$, for which the utility of player A,
$\mathcal{U}_A(.,.)$ is maximized. This is exactly the converse of
the procedure followed when best responses are studied. Indeed, in
the best response algorithm we don't presuppose that player A will
play some action, but we ask, given that player B will play an
action, say $b_k$, which action of player A maximizes his utility
function $\mathcal{U}_A(.,.)$.

Let us briefly recapitulate what we just described. In the best
response algorithm, we are searching player's A set of actions but
in the periodic actions algorithm described by the inequalities
(\ref{payoffreq}), we search the set of players B action, given that
A plays a specific action.

We also note that the periodicity number $n$ depends on the payoffs
rather than on the number of available actions.

\begin{theo}\label{1}
Every finite action simultaneous 2-player strategic form game
contains at least one periodic action.
\end{theo}

\begin{proof}
Since there are only finitely many actions, starting with any
  action $x_*$ of A and iteratively applying the operator
$\mathcal{Q}$ will eventually lead to some $x$ that has already
occurred before in the chain, that is, $\mathcal{Q}^n(x)=x$. But
then, because $\mathcal{Q}$ is defined by the inequalities
(\ref{payoffreq}), $\mathcal{Q}^{n+k}(x)=\mathcal{Q}^kx$ for every
positive integer $k$, and $x$ is periodic.
\end{proof}

This reasoning  reveals another property of the set of periodic
actions in finite action games.  We modify the definition of set
stability from Bernheim \cite{bernheim} as follows:

\begin{defi} [Set Stability]
Consider an automorphism $\mathcal{Q}:\mathcal{M}(A)\rightarrow
\mathcal{M}(A)$. In addition, let $A\subseteq A\cup B \subseteq
\mathcal{M}(A)$, with $A\cap B= \varnothing$. The set $A$ is set
stable under the action of the map $\mathcal{Q}$ if, for any initial
$x_0$ $\in$ $A\cup B$ and any sequence $x_k$ formed by taking
$x_{k+1}$ $\in$ $\mathcal{Q}(x_k)$, there exists $x_K$ $\in$ $A\cup
B$ such that $d(x_K,x^1)<\epsilon$, with $x^1$ $\in$ $A$. For finite
sets, this implies that any sequence formed by the act of the
operator $\mathcal{Q}$ on elements produces an element $x_k$ for any
initial $x_0$, with $x_k$ belonging to the set stable set $A$. An
analogous definition applies  for the set of actions of player B and
the operator $\mathcal{Q}':\mathcal{M}(B)\rightarrow
\mathcal{M}(B)$.
\end{defi}

\begin{theo}\label{2}

Let $\mathcal{P}(A)$ and $\mathcal{P}(B)$ denote the set of periodic
strategies for players A and B. The sets $\mathcal{P}(A)$ and
$\mathcal{P}(B)$ are set stable, under the action of the maps
$\mathcal{Q}$ and $\mathcal{Q}'$ respectively.

\end{theo}

\noindent The theorem says that for any  non-periodic action $x_0$,
we eventually arrive  in the periodicity cycle of some action $x_K$,
that is:
\begin{equation}\label{periodicitysimplex}
x_0{\,}{\,}{\,}\xrightarrow{P}{\,}{\,}{\,}
\mathcal{Q}(x_0){\,}{\,}{\,}\xrightarrow{P}{\,}{\,}{\,}
\mathcal{Q}^2(x_0){\,}{\,}{\,}\xrightarrow{P} \cdots
x_K\xrightarrow{P}{\,}{\,}{\,} \mathcal{Q}x_K \xrightarrow{P} \cdots
\xrightarrow{P}\mathcal{Q}^{n-1} x_K{\,}{\,}{\,}\xrightarrow{P} x_K
\end{equation}

\begin{proof} The proof of this theorem is contained in the proof of Theorem 1, so
we omit it.
\end{proof}

Finally, let us present two versions of the algorithm for finding
periodic strategies. Here, we do not assume that the payoff table is
non-degenerate. The first compact form of the algorithm is the
following:
\begin{algo}Consider a finite action simultaneous 2-player strategic
  form game. We find periodic solutions according to the following algorithm.\\
  Start with any actions $\xi_0 \in \mathcal{M}(A),\eta_0 \in \mathcal{M}(B)$ of A and B.\\
  Given
  $\xi_k,\eta_k$ for $k\ge 0$, determine $\eta_{k+1}\in \mathcal{M}(B)$ such that
  \bel{algo1}
  \mathcal{U}_A(\xi_k,\eta_{k+1})\ge  \mathcal{U}_A(\xi_k,y) \quad
  \forall y\in \mathcal{M}(B).
  \qe
  Given
  $\xi_k,\eta_{k+1}$, determine $\xi_{k+1}\in \mathcal{M}(A)$ such that
  \bel{algo2}
  \mathcal{U}_B(\xi_{k+1},\eta_{k+1})\ge  \mathcal{U}_B(x,\eta_{k+1}) \quad
  \forall x\in \mathcal{M}(A).
  \qe
  Repeat with $k+1$ in place of $k$.\\
  Stop with step $n$ when $\xi_n=\xi_\ell$ or $\eta_n=\eta_\ell$ for some $\ell
  <n$.
  \end{algo}

Here is a more detailed description of the algorithm: 
\begin{itemize}
  \item Start from Player A and his first action $a_1$. Seek in the
  strategy space of B, that is in the discrete set of all the
  available actions of B, namely $b_1,b_2,b_3,...,b_N$, and find the
  action $b_j$ for which, the corresponding payoff of player A
  $\mathcal{U}_A(a_1,b_j)$ is maximized, that is, the payoff $\mathcal{U}_A(a_1,b_j)$ is
  the largest among all payoffs $\mathcal{U}_A(a_1,b_i)$, with
  $i=1,2,3,...j-1,...j+1,...N$. So with the procedure we just
  described we have the chain of actions $a_1\to b_j$ for the
  moment. Note that this looks like a map from the strategy space of
  player A to the strategy space of player B, and in fact, this describes the formal  iteration step underlying 
  periodicity.

  \item For the action $b_j$ found in the previous step, now seek
  that action of player A, among the available actions 
  $a_1,a_2,a_3,...,a_N$, for which the payoff $\mathcal{U}_B(a_i,b_j)$  of player B for the action
  $b_j$ is maximized. Suppose this occurs for
  the action $a_k$. So the chain of actions now becomes $a_1\to b_j\to
  a_k$.

  \item For the action $a_k$ found in the previous step, now seek
  that  action of player B, among the available   actions
  $b_1,b_2,b_3,...,b_N$, for which the payoff $\mathcal{U}_A(a_k,b_i)$  of player A for the action
  $a_k$ is maximized. Suppose this occurs for
  the action $b_m$. So the chain of actions now becomes $a_1\to b_j\to
  a_k\to b_m$.

  \item ....

  \item ....

  \item The above procedure is repeated, alternating between A and B.  Once one of the previous selected actions reoccurs, say $a_0$, then this action $a_0$ is characterized as a periodic action,
  and the final chain of actions is $a_0\to b_j\to
  a_k\to b_m\to....\to a_1$ (with $b_j$ now being the response to $a_0$, etc.). The corresponding
  actions of player B then are also periodic.

  \item Once such a periodic action and the correponding chain have been found,  the algorithm comes to an end.

  \item The same procedure can be adopted for all the actions of
  player A and B.

\end{itemize}

\subsection{Periodic Nash and Rationalizable Strategies}

We are interested in  those Nash strategies that are at the same
time periodic actions. Suppose that the strategy set $(x^*,y^*)$
constitutes one of the Nash equilibria of a two player finite action
simultaneous move game. Then the actions $(x^*,y^*)$ are mutually
best responses for the two players. In order for a Nash strategy to
be a periodic strategy, the following conditions must hold true.

\begin{theo}\label{3}

In a 2-player finite action, simultaneous, strategic form game, a
Nash strategy $(x^*,y^*)$ of a game is periodic if
\begin{align}\label{nashconst}
&\varphi_1(x^*)=y^*\\ \notag & \varphi_2 (y^*)=x^*
\end{align}
with $\varphi_1$,$\varphi_2$ defined in such a way that the
inequalities (\ref{payoffreq}), (\ref{paytheplayer}) hold true. In
addition, the periodicity number for each action is equal to one,
that is $n=1$ and,
\begin{align}\label{nashconst1}
&\mathcal{Q}(x^*)=x^*\\ \notag & \mathcal{Q}'(y^*)=y^*
\end{align}

\end{theo}
\begin{proof}

The proof of Theorem \ref{3} is simple, but we must bear in mind
that the maps $\varphi_{1,2}$ do not yield in general the best
response sets of the players involved in a game. Suppose that for
the Nash strategy $(x^*,y^*)$, the relations (\ref{nashconst}) hold
true. Acting on the first with the map $\varphi_2$, and with
$\varphi_1$ on the second relation, we obtain the relations:
\begin{align}\label{nashconst2}
&\varphi_2\circ \varphi_1(x^*)=\varphi_2 (y^*)\\ \notag &
\varphi_1\circ \varphi_2 (y^*)=\varphi_1 (x^*)
\end{align}
Using relations (\ref{nashconst}),  the equations (\ref{nashconst2})
become:
\begin{align}\label{nashconst3}
&\phi_2\circ \phi_1(x^*)=x^*\\ \notag & \phi_1\circ \phi_2 (y^*)=y^*
\end{align}
Hence, the Nash actions $(x^*,y^*)$ are periodic. The relations
(\ref{nashconst3}) can be cast in terms of the operators
$\mathcal{Q}$ and $\mathcal{Q}'$ as
\begin{align}\label{nashconst4}
&\mathcal{Q}(x^*)=x^*\\ \notag & \mathcal{Q}' (y^*)=y^*
\end{align}
It is obvious that the periodicity number for the two actions is
$n=1$.
\end{proof}

Similarly, also rationalizable strategies can be periodic. This is
the case if  the rationalizability chain is identical to the
periodicity cycle. In particular, this is the case  if the
rationalizability chains of belief contain actions that satisfy at
every step the inequalities (\ref{payoffreq}) and
(\ref{paytheplayer}).

\subsection{Examples}
\subsubsection{Games with and without Periodic Nash Equilibria-Four Choices two Player Games}

We start  with Game 1A  in Table \ref{game1a}, which is an analogue
of one of the games in  \cite{bernheim}. We shall focus on the
choices of player A , but similar results hold for  B's actions.
Using the algorithm that the inequalities of relation
(\ref{payoffreq}) dictate, we can construct the periodicity cycles
\begin{align}\label{g1perc1}
&a_1{\,}{\,}{\,}\xrightarrow{P}{\,}{\,}{\,}
b_3{\,}{\,}{\,}\xrightarrow{P}{\,}{\,}{\,}
a_3{\,}{\,}{\,}\xrightarrow{P}{\,}{\,}{\,}
b_1{\,}{\,}{\,}\xrightarrow{P}{\,}{\,}{\,} a_1 \\ \notag
&a_3{\,}{\,}{\,}\xrightarrow{P}{\,}{\,}{\,}
b_1{\,}{\,}{\,}\xrightarrow{P}{\,}{\,}{\,}
a_1{\,}{\,}{\,}\xrightarrow{P}{\,}{\,}{\,}
b_3{\,}{\,}{\,}\xrightarrow{P}{\,}{\,}{\,} a_3
\end{align}

\begin{table}[h]
\centering
  \begin{tabular}{| l |l |l | l |l |}
    \hline
       &  $b_1$ & $b_2$ & $b_3$& $b_4$ \\ \hline
  $a_1$ & 0,7 & 2,5 & 7,0& 0,1  \\ \hline
    $a_2$ & 5,2 & 3,3& 5,2& 0,1\\
    \hline
$a_3$ & 7,0 & 2,5& 0,7& 0,1\\
    \hline
$a_4$ & 0,0 & 0,-2& 0,0& 10,-1\\
    \hline
 \end{tabular}
\caption{Game 1A} \label{game1a}
\end{table}
The periodicity number is $n=2$ for both actions, $a_1$ and $a_3$.
For the actions that constitute a Nash equilibrium, it is not
possible to construct such a cycle. Nevertheless, if we apply the
algorithm (\ref{payoffreq}), we obtain the following cycle:
\begin{align}\label{g1perc1}
a_2{\,}{\,}{\,}\xrightarrow{P}{\,} b_1{\,}\xrightarrow{P}{\,}
a_1{\,}\xrightarrow{P}{\,} b_3{\,}\xrightarrow{P}{\,}
a_3{\,}\xrightarrow{P}{\,} b_1{\,}\xrightarrow{P}{\,}
a_1{\,}{\,}{\,}
\end{align}
It is obvious that the cycle of the non-periodic Nash action $a_2$
merges into the periodic cycle of the periodic action $a_1$, as
predicted by Theorem \ref{2}. We next consider the rationalizability
cycles. The actions $a_1$ and $a_3$ are both rationalizable.
\begin{table}[h]
\centering
  \begin{tabular}{| l |l |l | l |l |}
    \hline
       &  $b_1$ & $b_2$ & $b_3$& $b_4$ \\ \hline
  $a_1$ & 0,7 & 2,5 & 7,0& 0,1  \\ \hline
    $a_2$ & 5,2 & 7,7& 5,2& 0,1\\
    \hline
$a_3$ & 7,0 & 2,5& 0,7& 0,1\\
    \hline
$a_4$ & 0,0 & 0,-2& 0,0& 10,-1\\
    \hline
 \end{tabular}
\caption{Game 1B} \label{game1b}
\end{table}
 and
periodic, and  the rationalizability cycles  coincide with the
periodicity cycles. Here, a rationalizability cycle is  a cycle
based on rationality, that is,  acting optimally under some beliefs
about the opponents actions. Indeed, such a cycle  looks like:
\begin{align}\label{ccfgbfg123}
&a_1{\,}{\,}{\,}\xrightarrow{R}{\,}{\,}{\,}
b_3{\,}{\,}{\,}\xrightarrow{R}{\,}{\,}{\,}
a_3{\,}{\,}{\,}\xrightarrow{R}{\,}{\,}{\,}
b_1{\,}{\,}{\,}\xrightarrow{R}{\,}{\,}{\,} a_1\\ \notag &
a_3{\,}{\,}{\,}\xrightarrow{R}{\,}{\,}{\,}
b_1{\,}{\,}{\,}\xrightarrow{R}{\,}{\,}{\,}
a_1{\,}{\,}{\,}\xrightarrow{R}{\,}{\,}{\,}
b_3{\,}{\,}{\,}\xrightarrow{R}{\,}{\,}{\,} a_3
\end{align}
The reasoning behind this cycle is based on this system of beliefs:
Player A considers action $a_1$ rational if he believes that player
B will play $b_3$, which is rational for player B if he believes
that player A will play $a_3$. Accordingly, A will consider playing
$a_3$ rational if he believes that player B will play $b_1$, which
would be rational for player B if he believes that player A will
play $a_1$. Therefore, we obtain a cycle of rationalizability based
on pure utility maximization rationality.

The Nash action $a_2$  is not contained in  such a cycle due to the
fact that A will be forced to play $a_2$, and B would never play
$b_1$ or $b_3$ as a best response to $a_2$. So, the Nash strategy is
``forced`` to be rationalizable. In this game, the non-Nash
rationalizable actions are periodic actions which actually are the
only periodic strategies and also the rationality cycles and
periodicity cycles coincide.

\noindent We now slightly modify Game 1A and consider Game 1B
 in Table \ref{game1b}. The difference is that the Nash
equilibrium payoffs are changed. In this case, the periodicity
cycles of the actions $a_1$ and $a_3$ remain intact, but in this
case, the Nash action $a_2$ is also periodic, with periodicity
cycle:
\begin{align}\label{ccfgbfg1uhy23}
a_2{\,}{\,}{\,}\xrightarrow{R}{\,}{\,}{\,}
b_2{\,}{\,}{\,}\xrightarrow{R}{\,}{\,}{\,} a_2
\end{align}
Again, the periodicity and rationalizability cycles for the Nash
action $a_2$ coincide.

\subsubsection{$2\times 2$ Games }

Now we analyze $2\times 2$ simultaneous strategic form games.
Consider first Game 2  in Table \ref{game2}.
\begin{table}[h]
\centering
  \begin{tabular}{| l |l |l | }
    \hline
       &  $b_1$ & $b_2$ \\ \hline
  $a_1$ & 3,5 & 0,2   \\ \hline
    $a_2$ & 4,3 & 5,4\\
    \hline
 \end{tabular}
\caption{Game 2}\label{game2}
\end{table}
The Nash equilibrium consists of the actions $(a_2,b_2)$. Following
the reasoning of relation (\ref{payoffreq}), we can construct the
following periodicity cycles:
\begin{align}\label{g1perc}
&a_1{\,}{\,}{\,}\xrightarrow{P}{\,}{\,}{\,}
b_1{\,}{\,}{\,}\xrightarrow{P}{\,}{\,}{\,} a_1 \\ \notag &
a_2{\,}{\,}{\,}\xrightarrow{P}{\,}{\,}{\,}
b_2{\,}{\,}{\,}\xrightarrow{P}{\,}{\,}{\,} a_2
\end{align}
Obviously, all actions have a periodicity cycle and additionally all
the periodicity numbers are equal to one in this particular game.
Note that the actions that enter the Nash equilibrium are also
periodic. However, the action $a_1$ is strictly dominated by the
action $a_2$ for all cases, so it is not rationalizable. So we can
never construct a cycle based on rationality argument for this
action. Indeed, player A would never consider the action $a_1$ to be
a rational move, since it can never be a best response.

Nevertheless, we can construct a cycle based on rationality
arguments for the $a_2$ action. Indeed, player A would consider
$a_2$ to be a rational move if he believed that player B would play
$b_2$, which would be rational for player B if he believes that
player A plays $a_2$. According to this line of reasoning we can
construct the following rationalizability cycle,
 \begin{align}\label{rc}
&a_2{\,}{\,}{\,}\xrightarrow{R}{\,}{\,}{\,}
b_2{\,}{\,}{\,}\xrightarrow{R}{\,}{\,}{\,} a_2.
\end{align}
In this particular game, the set of periodic actions for player $A$
consists of both actions $a_1$ and $a_2$, that is
$\mathcal{P}(A)=\{a_1,a_2\}$, while the set of rationalizable
actions that are not Nash actions is empty. The set of Nash actions
consists of the action $\{a_2\}$. For this particular game, the Nash
equilibrium happens to be periodic.

In this game, the iterated elimination of dominated strategies
results in $(a_2,b_2)$ which is the Nash equilibrium. This class of
games describes competition between two firms that choose quantities
that they produce, knowing that the total quantity that is made
available in the market actually determines the price
\cite{gamesofstrategy}. In this game,  a periodic Nash equilibrium
is the only action that remains after the iterated elimination of
dominated strategies.

\subsubsection{Some Standard Games}

Before closing this section, we study the periodicity properties of
the players available actions for  the prisoner's dilemma game, the
battle of sexes game and finally the matching pennies game. Let us
start with the prisoner's dilemma game,  Game 3
 in Table \ref{pd},
\begin{table}[h]
\centering
  \begin{tabular}{| l |l |l | }
    \hline
       &  $b_1$ & $b_2$ \\ \hline
  $a_1$ & b,b & d,a   \\ \hline
    $a_2$ & a,d & c,c\\
    \hline
 \end{tabular}
\caption{Game 3, Prisoner's Dilemma} \label{pd}
\end{table}
with  $a<b<c<d$. The action $a_1$ is rationalizable but not
periodic. The action $a_2$ is periodic and the strategy $(a_2,b_2)$
contains periodic actions. Actually, the periodicity cycle in this
case is
\begin{equation}\label{perprisdlm}
a_2{\,}{\,}{\,}\xrightarrow{P}{\,}{\,}{\,}
b_2{\,}{\,}{\,}\xrightarrow{P}{\,}{\,}{\,}
a_1{\,}{\,}{\,}\xrightarrow{P}{\,}{\,}{\,}
b_1{\,}{\,}{\,}\xrightarrow{P}{\,}{\,}{\,} a_2
\end{equation}
Note that the periodicity number is $n=2$ in this case.

Let us continue with the Battle of Sexes,  Game 4  in Table
\ref{bs}.
\begin{table}[h]
\centering
  \begin{tabular}{| l |l |l | }
    \hline
       &  $b_1$ & $b_2$ \\ \hline
  $a_1$ & 2,1 & 0,0   \\ \hline
    $a_2$ & 0,0 & 1,2\\
    \hline
 \end{tabular}
\caption{Game 4, Battle of Sexes} \label{bs}
\end{table}
There are two Nash equilibria,  $(a_1,b_1)$ and $(a_2,b_2)$,  and
both actions $a_1$ and $a_2$ are periodic and rationalizable. There
are no non-Nash strategies that are rationalizable. In this game, we
always have $n=1$.

Finally, in the matching pennies game,  Game 5
 in Table \ref{pcg}.
\begin{table}[h]
\centering
  \begin{tabular}{| l |l |l | }
    \hline
       &  $b_1$ & $b_2$ \\ \hline
  $a_1$ & 1,-1 & -1,1   \\ \hline
    $a_2$ & -1,1 & 1,-1\\
    \hline
 \end{tabular}
\caption{Game 5, The Matching Pennies Game} \label{pcg}
\end{table}
we have  $n=2$ and the actions $a_1,a_2$ are both periodic and
rationalizable, that is, we can construct the following cycles:
\begin{align}\label{g1perc}
&a_1{\,}{\,}{\,}\xrightarrow{P}{\,}{\,}{\,}
b_1{\,}{\,}{\,}\xrightarrow{P}{\,}{\,}{\,} a_2
{\,}{\,}{\,}\xrightarrow{P}{\,}{\,}{\,}
b_2{\,}{\,}{\,}\xrightarrow{P}{\,}{\,}{\,} a_1\\ \notag &
a_2{\,}{\,}{\,}\xrightarrow{R}{\,}{\,}{\,}
b_2{\,}{\,}{\,}\xrightarrow{R}{\,}{\,}{\,} a_1
{\,}{\,}{\,}\xrightarrow{R}{\,}{\,}{\,}
b_1{\,}{\,}{\,}\xrightarrow{R}{\,}{\,}{\,} a_2
\end{align}
All the actions of both players are periodic and  rationalizable.
There is no pure strategy Nash equilibrium. This motivates to turn
to  mixed strategies. This will be the subject of the next section.

Before closing this section, let us briefly comment on the case of
extensive form games and periodicity. Since every perfect
information extensive form game has a strategic form game
representation, all arguments apply to extensive form games. The
difference is that the strategic form representation of an extensive
form game has many degeneracies, so we may have many periodicity
cycles corresponding to a specific action. Here, however, this will
not be analyzed further.

\section{Periodicity, Rationalizability and Mixed Strategies in Finite Action Simultaneous Strategic Form Games}

\subsection{Essential Features of Periodicity in the Case of Mixed Strategies: Introductory Remarks}

We shall now develop   the periodicity concept for  mixed strategies
for finite action simultaneous strategic form games (again for
2-player games only), and derive its consequences. For some classes
of well-known  games, the periodic mixed strategies yield the same
payoff as the Nash strategies do. Sometimes, they yield even higher
payoffs.  The important difference, or advantage of such periodic
strategies is that, in contrast to the situation for Nash
equilibria, the payoff of a player does not depend on the action of
the opponent. For Nash,  the underlying rationality assumption is
crucial. Each player is not only rational himself, but assumes that
the opponent is rational as well, in the sense that he adopts the
best response to the own action. An equilibrium where both act
rationally in that sense is a Nash equilibrium.  For the periodicity
concept, in contrast, one assumes  that the opponent chooses the
best action not for himself, but for oneself. As we shall see, this
is as self-consistent as the Nash rationality assumption.

Let us now present the essential idea of the algorithm for mixed
strategies, again for simultaneous mixed strategy, two-player,
strategic form games with perfect information. And for simplicity,
we grant each player two actions only.

A  general mixed strategy for the player A then is of the form
\begin{equation}\label{mixa}
x_{\sigma}=pa_1+(1-p)a_2,
\end{equation}
and one for B looks like
\begin{equation}\label{mixb}
y_{\sigma}=qb_1+(1-q)b_2.
\end{equation}
The crucial parameters here are the probabilities $0\leq p,q\leq 1$.
We assume that the payoffs depend differentiably on them. In fact,
in game theory, it is usually assumed that the payoff for a mixed
strategy is the convex combination of those for pure strategies,
that is,
\begin{equation}
  \label{mixpay}
  \mathcal{U}_i(p,q)=pq\mathcal{U}_i(a_1,b_1)+(1-p)q\mathcal{U}_i(a_2,b_1)+p(1-q)\mathcal{U}_i(a_1,b_2)+(1-p)(1-q)\mathcal{U}_i(a_2,b_2) \text{ for } i=A,B.
\end{equation}
The difference between periodic and Nash mixed strategies then
reduces to the fact that for Nash, a player, say A,  optimizes
w.r.t. his own action and therefore, $\mathcal{U}_A(p,q)$ has to be
differentiated w.r.t. to his own probability $p$ to find the
optimum. In contrast, for the periodic strategy, A will need to
differentiate $\mathcal{U}_A(p,q)$ w.r.t. the opponent's probability
$q$. Since $p$ and $q$ enter linearly in \eqref{mixpay}, when we
differentiate w.r.t. $p$ and put the resulting expression $=0$ to
find the extremum, the result will only depend on $q$ and not on
$p$. That is, for Nash, the result is independent of the own action
and only depends on the opponent's choice of $q$. In contrast, when
we differentiate w.r.t. $q$ and set the result $=0$, it will no
longer depend on the opponent's $q$, but only on the own $p$.

We can thus turn the procedure into algorithmic form by iteratively
performing such an optimization for the 2 players in turn. We work
with the set-up defined in the beginning of this section, in
particular the utility functions \eqref{mixpay}.

\begin{itemize}

 \item  Start with some value of $p$ for player A and put $\frac{\partial \mathcal{U}_A(p,q)}{\partial q}=0$. The solution  specifies a mixed strategy action
  $p_p^*$ (with the subscript $p$ standing for ``periodic''.)

 \item  Continue with some value $q$ for player B and put $\frac{\partial \mathcal{U}_B(p,q)}{\partial p}=0$, leading to a mixed strategy action
  $q_p^*$.
  \item Check for both players if the following conditions hold true:
\begin{align}\label{periodicpayoff}
&\frac{\partial {\mathcal{U}_A}(p,q)}{\partial q}\Big{\lvert
}_{p=p_p^*}=0 & \\ \notag &{\mathcal{U}_A}_{p_p^*,q}=
\mathrm{max}\Big{(}{\mathcal{U}_A}(p,q)\Big{)},{\,}{\,}\forall
{\,}{\,}p,q
\end{align}
\begin{align}\label{periodicpayoff}
&\frac{\partial {\mathcal{U}_B}(p,q)}{\partial q}\Big{\lvert
}_{q=q_p^*}=0 & \\ \notag &{\mathcal{U}_B}(p,q_p^*)=
\mathrm{max}\Big{(}{\mathcal{U}_B}(p,q)\Big{)},{\,}{\,}\forall
{\,}{\,}p,q
\end{align}

\item Check if $q_p^*$ is the only value of
$q$ that maximizes ${\mathcal{U}_A}(p_p^*,q)$ and also check if
$p_p^*$ is the only value of $p$ that maximizes
${\mathcal{U}_B}(p,q_p^*)$.

  \item If so, stop; otherwise return to the beginning.
\end{itemize}
When the algorithm stops, we have found $p_p^*$ and $q_p^*$ which
are mutually optimal. That means that $q_p^*$ is that action of B
that is optimal for A when he plays $p_p^*$, and conversely.

\subsection{Periodicity for Mixed Strategies}

We now look for  periodicity patterns in  $2\times 2$ games in the
context of mixed strategies. As in the pure strategy case, this
periodicity will be materialized by two maps $\Phi_1$, $\Phi_2$ that
constitute the automorphisms $\mathcal{Q}=\Phi_2\circ \Phi_1$ and
$\mathcal{Q}'=\Phi_1\circ \Phi_2$. Their definition will be
different from  the pure strategy case. Take for example player A:
The operator ${\mathcal{Q}}$ has the property that there exists a
positive integer $n$ and some action $x_{\sigma} \in \Delta
(\mathcal{M}(A))$,  for which $\mathcal{Q}^nx_{\sigma}=x_{\sigma}$.
The actions of the maps $\Phi_1$ and $\Phi_2$ are defined in the
mixed strategies case as follows,
 \begin{align}
 &\Phi_1: \Delta (\mathcal{M}(A))\rightarrow \Delta (\mathcal{M}(B))\\ \notag
  &\Phi_2: \Delta (\mathcal{M}(B))\rightarrow \Delta (\mathcal{M}(A))
 \end{align}
where  $\mathcal{M}(A)$ and $\mathcal{M}(B)$ are the available
strategies space of player A and B, and  $\Delta (\mathcal{M}(A))$
and $\Delta (\mathcal{M}(B))$ are the probability distributions
over the corresponding strategy spaces. These two maps $\Phi_1$ and
$\Phi_2$ are defined such that  we have for all $k\ge 1$
\begin{align}\label{payoffreqmixed}\notag
&\mathcal{U}_{B_{p,q}}((\Phi_2 \circ
\Phi_1)^k(x_{\sigma}),\Phi_1\circ (\Phi_2 \circ
\Phi_1)^{k-1}(x_{\sigma}))>
\mathcal{U}_{B_{p,q}}(x_{{\sigma}_1},\Phi_1\circ (\Phi_2 \circ
\Phi_1)^{k-1}(x_{\sigma}))
\\ \notag
&\forall {\,}{\,}x_{{\sigma}_1}{\,}\in {\,}\Delta
(\mathcal{M}(A)\backslash\{(\Phi_2 \circ \Phi_1)^k(x_{\sigma})\})
\\ \notag
&\mathcal{U}_{A_{p,q}}((\Phi_2 \circ
\Phi_1)^k(x_{\sigma}),\Phi_1\circ (\Phi_2 \circ
\Phi_1)^{k}(x_{\sigma}))> \mathcal{U}_{A_{p,q}}((\Phi_2 \circ
\Phi_1)^kx_{\sigma},y_{{\sigma}_1})
\\  & \forall
{\,}{\,}y_{{\sigma}_1}{\,}\in {\,}\Delta
(\mathcal{M}(B)\backslash\{\Phi_1\circ (\Phi_2 \circ
\Phi_1)^{k}(x_{\sigma})\})
\end{align}
when we consider player A.   The algorithm implied by the
inequalities (\ref{payoffreqmixed}) dictates that starting with a
mixed strategy of player A, namely $x_{\sigma}$, and upon which we
act with the map $\Phi_1$, we search in player's B set of
probability distributions $\Delta (\mathcal{M}(B))$, in order to
find which mixed strategy maximizes the expected utility of player
A. This then is iterated.  Accordingly, just like in the pure
strategy case, it is possible that the  process returns to the
initial mixed strategy, $x_{\sigma}$. In that case, there is  a
chain of mixed strategies of the following form:
\begin{equation}\label{periodicitysimplex}
x_{\sigma}{\,}{\,}{\,}\xrightarrow{P}{\,}{\,}{\,}
\Phi_1(x_{\sigma}){\,}{\,}{\,}\xrightarrow{P}{\,}{\,}{\,}
\Phi_2\circ \Phi_1 (x_{\sigma}){\,}{\,}{\,}\xrightarrow{P} \cdots
\xrightarrow{P}{\,}{\,}{\,} x_{\sigma}
\end{equation}
where as in the pure strategy case, the letter $P$ denotes the
procedure described in relation (\ref{payoffreqmixed}) above. A
mixed strategies for which we can find such a chain, is called
periodic, and as in the pure strategy case, such strategies  satisfy
\begin{equation}\label{mixrer}
\mathcal{Q}^nx_{\sigma}=x_{\sigma}
\end{equation}
The last inequality of  (\ref{payoffreqmixed}) then becomes
\begin{equation}\label{lastineq}
\mathcal{U}_{A_{p,q}}({\mathcal{Q}}^n(x_{\sigma}),\Phi_1\circ
{\mathcal{Q}}^{n}(x_{\sigma}))>
\mathcal{U}_{A_{p,q}}({\mathcal{Q}}^{n}(x_{\sigma}),y_{{\sigma}_1})
\forall {\,}{\,}y_{{\sigma}_1}{\,}\in
{\,}\Delta\Big{(}\mathcal{M}(B)\backslash\{\Phi_1\circ
{\mathcal{Q}}^{n}(x_{\sigma})\})\Big{)}
\end{equation}

Dealing  with mixed strategies provides the  advantage that we can
use differential calculus to identify the optimizers at each step.
The action of the map $\Phi_1$ on $x_{\sigma}$ is equivalent to the
maximization of $\mathcal{U}_{A_{p,q}}$ with respect to $q$. Indeed,
take for example the initial inequality of relation
(\ref{payoffreqmixed}). The map $\Phi_1$ yields a strategy $\in$
$\Delta (\mathrm{N}(B))$ by  which the expected utility of player A
is maximized. Hence, if we differentiate $\mathcal{U}_{A_{p,q}}$
with respect to $q$, the corresponding solution $p_p^*$ is equal to
$\Phi_1 (x_{\sigma})$,
\begin{equation}\label{maxcondition}
p_p^*=\Phi_1 (x_{\sigma})
\end{equation}
We shall use the property implied by relation (\ref{maxcondition})
in one of the next subsections to bring out interesting features in
some classes of games. The same considerations
 apply for player B. Thereby, the
corresponding inequalities (\ref{payoffreqmixed}) for a given
initial mixed strategy $y_{\sigma}$ $\in$ $\Delta (\mathcal{M}(A))$,
now become
\begin{align}\label{paytheplayermixed2}
&\mathcal{U}_{A_{p,q}}((\Phi_1 \circ
\Phi_2)^k(y_{\sigma}),\Phi_2\circ (\Phi_1 \circ
\Phi_2)^{k-1}(y_{\sigma}))>
\mathcal{U}_{A_{p,q}}(y_{{\sigma}_1},\Phi_2\circ (\Phi_1 \circ
\Phi_2)^{k-1}(y_{\sigma}))
\\ \notag &\forall {\,}{\,}y_{{\sigma}_1}{\,}\in {\,}\Delta (\mathcal{M}(B))\backslash\{\Phi_1 \circ \Phi_2)^k(y_{\sigma})\}
\\ \notag
&\mathcal{U}_{B_{p,q}}((\Phi_1 \circ
\Phi_2)^k(y_{\sigma}),\Phi_2\circ (\Phi_1 \circ
\Phi_2)^{k}(y_{\sigma}))> \mathcal{U}_{B_{p,q}}((\Phi_1 \circ
\Phi_2)^{k}(y_{\sigma}),x_{{\sigma}_1})
\\ \notag & \forall {\,}{\,}x_{{\sigma}_1}{\,}\in {\,} \Delta ( \mathcal{M}(A))\backslash\{\Phi_2\circ (\Phi_1 \circ \Phi_2)^{k}(y_{\sigma})\}
\end{align}
for every $k\ge 1$. The  inequality at the periodic value $k=n$ can
be written in terms of the operator $\mathcal{Q}'$ as
\begin{equation}\label{lastineq}
\mathcal{U}_{B_{p,q}}({\mathcal{Q}'}^n(y),\Phi_2\circ
{\mathcal{Q}'}^{n}(y))>
\mathcal{U}_{B_{p,q}}({\mathcal{Q}'}^{n}(y),x_{{\sigma}_1}) \forall
{\,}{\,}x_{{\sigma}_1}{\,}\in {\,}\mathcal{M}(A)
\end{equation}
where the  operator $\mathcal{Q}'$ is given by
\begin{equation}\label{oper2}
\mathcal{Q}'=\Phi_1 \circ \Phi_2 .
\end{equation}
Hence, a periodic action $y_{\sigma}$ of player B satisfies
\begin{equation}\label{mixrer1}
{\mathcal{Q}'}^ny_{\sigma}=y_{\sigma}
\end{equation}
We shall now discuss some standard  strategic form games for which
the periodicity argument applies.

\subsection{Periodic Mixed Strategies in Some  Games}

To see which games have the properties  discussed in the previous
sections, we look at  general characteristics of the payoff matrix.
The maximization of  A's  utility function ${\mathcal{U}_A}_{p,q}$
with respect to $q$ yields
\begin{align}\label{condgen}
p_p^*=
\frac{{\mathcal{U}_A}_{p,q}(a_2,b_2)-{\mathcal{U}_A}_{p,q}(a_2,b_1)}{({\mathcal{U}_A}_{p,q}(a_1,b_1)+{\mathcal{U}_A}_{p,q}(a_2,b_2)-{\mathcal{U}_A}_{p,q}(a_1,b_2)-{\mathcal{U}_A}_{p,q}(a_2,b_1))}
\end{align}
while the maximization with respect to $p$ yields
\begin{align}\label{conditionmax}
q_N^*=
\frac{\mathcal{U}_A(a_2,b_2)-\mathcal{U}_A(a_1,b_2)}{(\mathcal{U}_A(a_1,b_1)+\mathcal{U}_A(a_2,b_2)-\mathcal{U}_A(a_1,b_2)-\mathcal{U}_A(a_2,b_1))}\,
.
\end{align}
The relation (\ref{condgen}) yields the potential mixed periodic
strategy, while relation (\ref{conditionmax}) yields the potential
Nash mixed strategy.

Now, we will exploit the fact that when the opponent plays a mixed
Nash equilibrium strategy, the player's expected utility is
independent of his own randomization over his own strategies, and at
the same time, the utility is maximized. For example, in the case of
player A, this would mean that the corresponding Nash equilibrium
does not depend on $p$ but on $q$. Hence, we can built games in such
a way that the mixed periodic strategies are related to the mixed
Nash equilibria and then explore the consequences of such a
relation.

\subsubsection{First Type of Games}

We take a look at  $2\times 2$ games that  satisfy
\begin{align}\label{ikmplg2}
& p_p^*=q_N^* \\ \notag & q_p^*=p_N^*
\end{align}
where $p_p^*$ and $p_N^*$ are the mixed periodic and mixed Nash
equilibria for player A and $q_p^*$ and $q_N^*$ are those for player
B. Hence, it is obvious how the robustness of the corresponding
expected utilities is achieved. Making use of relations
(\ref{conditionmax}) and (\ref{condgen}), relations (\ref{ikmplg2})
impose some restrictions on the payoff matrices, which are the
following,
 \begin{equation}\label{rest4}
{\mathcal{U}_A}_{p,q}(a_1,b_2)={\mathcal{U}_A}_{p,q}(a_2,b_1),{\,}{\,}{\,}{\mathcal{U}_B}_{p,q}(a_1,b_2)={\mathcal{U}_B}_{p,q}(a_2,b_1)
\end{equation}
Let us illustrate this result for the {\it Battle of Sexes} in Table
\ref{game1m}.
\begin{table}[h]
\centering
  \begin{tabular}{| l |l |l | }
    \hline
       &  $b_1$ & $b_2$ \\ \hline
  $a_1$ & 2,1 & 0,0   \\ \hline
    $a_2$ & 0,0 & 1,2\\
    \hline
 \end{tabular}
\caption{Mixed Strategies Game 1}\label{game1m}
\end{table}
The mixed Nash equilibrium for this game is
$(p_N^*=\frac{2}{3},q_N^*=\frac{1}{3})$. If we maximize A's expected
utility w.r.t. $q$, we obtain $\frac{\partial
{\mathcal{U}_A}_{p,q}}{\partial q}=-1+3p$, hence the mixed periodic
strategy is $p_p^*=1/3$, and again, this is independent  of the
value of $q$. The corresponding  maximization procedure for player B
yields the mixed periodic strategy $q_p^*=2/3$. Let us now examine
the expected utilities of the players. The expected utility for
player A for the mixed ''periodic'' (we shall use the term periodic
even though these strategies are not periodic per se, but these
strategies result from using the first three steps of the mixed
strategies algorithm we described in the previous sections) strategy
$p_p^*=1/3$ is
\begin{equation}\label{mix1}
{\mathcal{U}_A}_{p,q}(p_p^*=1/3,q)=\frac{2}{3}
\end{equation}
and is independent of $q$. By symmetry, the same applies for players
B utility for $q_p^*=2/3$:
\begin{equation}\label{mix12}
{\mathcal{U}_B}_{p,q}(p,q_p^*=2/3)=\frac{2}{3}.
\end{equation}
The expected utilities  of  the players when the opponent plays his
mixed Nash strategy $p_N^*=q_N^*=1/3$ are
\begin{equation}\label{mix1}
{\mathcal{U}_A}_{p,q}(p,q_N^*=1/3)={\mathcal{U}_B}_{p,q}(p_N^*=2/3,q)=\frac{2}{3},
\end{equation}
and these values do not depend on the own action, and the expected
utilities are maximized when the \textit{opponent} plays mixed Nash.
Also notice that the expected utilities for the mixed periodic
strategy take also their \textit{maximum values}. The utilities
corresponding to the periodic strategies and to the mixed Nash
equilibria coincide.

The disadvantage of the mixed Nash strategy, compared  to the mixed
periodic strategy, is that in order for the expected utility to be
maximized, the \textit{opponent} has to play Nash. This renders all
the strategies of the player optimal. However, this does not happen
in the periodic mixed strategy case, where when a player plays his
own periodic mixed strategy, the expected utility is maximized,
\textit{regardless} of what the other player plays.

We observe furthermore that if player A plays  his own Nash mixed
strategy $p_N^*=2/3$, his expected utility is
\begin{equation}\label{mix1}
{\mathcal{U}_A}_{p,q}(p_N^*=2/3,q)=\frac{1}{3}+q.
\end{equation}
Accordingly the expected utility of player B for $q_N^*=1/3$ is
\begin{equation}\label{mix12}
{\mathcal{U}_B}_{p,q}(p,q_N^*=1/3)=\frac{4}{3}-p
\end{equation}
These utilities \textit{depend} on what the other player plays, and
thus the mixed strategies of each player \textit{do not} render the
corresponding payoff robust against  the opponent's strategies.

In contrast, the mixed periodic strategies render the corresponding
payoffs \textit{robust} to what the opponent chooses, and in
addition these \textit{maximize} the expected utility functions.
This result is very intriguing, since the mixed periodic strategies
we found, namely $(p_p^*=1/3,q_p^*=2/3)$ are not rationalizable
actions. Nevertheless, we have seen  that the player who adopts
these strategies always achieves equal or larger payoff in
comparison to the mixed Nash payoff, \textit{regardless} of what his
opponent eventually plays (we assume here $0< p,q <1$).

\subsubsection{Second Type of Games}

We now turn to games that satisfy
\begin{align}\label{ikjfgjmplg2}
& p_p^*=1-q_N^* \\ \notag & q_p^*=1-p_N^*.
\end{align}
These conditions render the corresponding expected utilities robust
against the opponent strategies. In addition, condition
(\ref{ikjfgjmplg2}) implies
 \begin{equation}\label{refgjfjst4}
{\mathcal{U}_A}_{p,q}(a_1,b_1)={\mathcal{U}_A}_{p,q}(a_2,b_2),{\,}{\,}{\,}{\mathcal{U}_B}_{p,q}(a_1,b_1)={\mathcal{U}_B}_{p,q}(a_2,b_2)
\end{equation}
Let us illustrate this result for Game 2  in Table \ref{game2mm}.
\begin{table}[h]
\centering
  \begin{tabular}{| l |l |l | }
    \hline
       &  $b_1$ & $b_2$ \\ \hline
  $a_1$ & 2,5 & 50,6   \\ \hline
    $a_2$ & 3,10 & 2,5\\
    \hline
 \end{tabular}
\caption{Mixed Strategies Game 2}\label{game2mm}
\end{table}
We have  two pure  Nash equilibria,  $(a_1,b_2)$ and $(a_2,b_1)$,
which at the same time are periodic, and the  mixed Nash equilibrium
$(p_N^*=\frac{5}{6},q_N^*=\frac{48}{49})$. The mixed periodic
strategy of A is $p_p^*=1/49$ according to  (\ref{ikjfgjmplg2}), and
that for B is  $q_p^*=1/6$.  The expected utility for player A at
$p_p^*=1/49$ is
\begin{equation}\label{mjfgjix1}
{\mathcal{U}_A}_{p,q}(p_p^*=1/49,q)=\frac{146}{49}
\end{equation}
which  is again independent of $q$. Similarly, that of B  for
$q_p^*=1/6$ is
\begin{equation}\label{mjfgjix12}
{\mathcal{U}_B}_{p,q}(p,q_p^*=1/6)=\frac{35}{6}.
\end{equation}

Again, these coincide with the expected utilities resulting from
playing Nash,
\begin{equation}\label{mixjfgj1}
{\mathcal{U}_A}_{p,q}(p,q_N^*=48/49)=\frac{146}{49},\quad
{\mathcal{U}_B}_{p,q}(p_N^*=5/6,q)=\frac{35}{6}.
\end{equation}
Again, the mixed Nash payoffs are independent of the own action and
the expected utilities are maximized when the \textit{opponent}
plays his mixed Nash strategy. And again,  the expected utilities
for the mixed periodic strategy take their maximum values, which are
equal to the ones obtained for the Nash strategies. However, if
player A plays for instance his own Nash mixed strategy $p_N^*=5/6$,
his expected utility is
\begin{equation}\label{mix1}
{\mathcal{U}_A}_{p,q}(p_N^*=5/6,q)=-\frac{239}{6}q+42
\end{equation}
while the expected utility for player B, when he plays
$q_N^*=48/49$, is
\begin{equation}\label{mix12}
{\mathcal{U}_B}_{p,q}(p,q_N^*=1/3)=\frac{485-239p}{49}
\end{equation}
Again, these   utilities depend on what the \textit{opponent} plays,
in contrast to the  mixed periodic strategies. And again, the two
types of utilities coincide.

As a general remark for the types of games appearing in Tables
\ref{game1m} and \ref{game2mm}, we should note that
$U_i(p_p,q_p)=U_i(p_N,q_N)$, whenever both the Nash and the
periodicity solution concepts require a mixed strategy. This can
also be seen by looking Eqs. (\ref{mjfgjix1}), (\ref{mjfgjix12}) and
(\ref{mixjfgj1}).

Since in game theory, one is often more interested in pure than in mixed solutions, we now look at another game in Table \ref{game2mmnewgame},
\begin{table}[h]
\centering
  \begin{tabular}{| l |l |l | }
    \hline
       &  $b_1$ & $b_2$ \\ \hline
  $a_1$ & 4,0 & -1,-1   \\ \hline
    $a_2$ & -3,1 & 0,4\\
    \hline
 \end{tabular}
\caption{Alternative Example}\label{game2mmnewgame}
\end{table}
For this game, by applying the periodicity algorithm, we find that
the mixed periodic payoffs are
\begin{equation}\label{solutionnewgame1}
{\mathcal{U}_A}_{p,q}(p_p^*=1/2,q_p^*=1/2)=\frac{1}{2},\,\,\,{\mathcal{U}_B}_{p,q}(p_p^*=1/2,q_p^*=1/2)=\frac{3}{2}\,
,
\end{equation}
and the corresponding mixed Nash payoffs are
\begin{equation}\label{solutionnewgame1}
{\mathcal{U}_A}_{p,q}(p_N^*=1/2,q_N^*=1/2)=\frac{1}{2},\,\,\,{\mathcal{U}_B}_{p,q}(p_N^*=1/2,q_N^*=1/2)=\frac{3}{2}\,
,
\end{equation}
so the two solution concepts yield the same mixed payoffs.
However, the pure Nash equilibrium is $(a_2,b_2)$, so the player $B$
gets a higher payoff for the pure Nash strategy. The same would
apply for player A if we transposed the payoffs. Therefore,
sometimes periodicity is better, while at other times Nash is better. This
seems to be a generic characterization which covers more types of
games, compared with the  first and second type games  discussed earlier in
this section.

We now turn to games where  the expected periodic utilities are
actually higher than those from mixed Nash.

\subsubsection{When Periodicity is Better
than Nash Mixed Strategies: The Case of Collective Action Games and
Prisoner Dilemma Type of Games} In this section, we shall treat a
class of games that includes the  ``Pure Public Good'' games. Let us
recall an example. Suppose that two  oil companies want to extract
oil near some island and transfer it to international markets. For
that, a pipeline is needed  The government's public policy allows
only one pipeline, so both companies must share the pipeline, when
it is constructed. Both companies will benefit from it, but  the
question is who is going to fund the construction of this pipeline.

Such games are inherent to problems of collective action
\cite{gamesofstrategy}. In this kind of games, the actions that
yield better payoffs for the players do not belong to the set of
best private interest actions of the players, or more formally, the
Pareto optimal outcome is not necessarily the Nash equilibrium.

\noindent The pipeline project, like any Pure Public Good game,
has two characteristic properties that the benefits  are {\it non
excludable} and  {\it non-rival}. Such a game can be represented in
matrix form in Game 3  in Table \ref{game3m} (taken from
\cite{gamesofstrategy}).
\begin{table}[h]
\centering
  \begin{tabular}{| l |l |l | }
    \hline
       &  $b_1$ & $b_2$ \\ \hline
  $a_1$ & 4,4 & -1,6  \\ \hline
    $a_2$ & 6,-1 & 0,0\\
    \hline
 \end{tabular}
\caption{Mixed Strategies Game 3}\label{game3m}
\end{table}
The Nash equilibrium is  $(a_2,b_2)$. The payoffs depend on the
quality and the time that it takes to materialize the project.
Obviously, the optimal action for both players is not to
participate, no matter what the other player does, that is, to act
as a ``free rider''. In contrast, the Pareto  optimum is achieved
when the strategy $(a_1,b_1)$ is adopted by both players.

The social optimal is always achieved when the total sum of the
players payoffs is maximized, but this requires a
\textit{cooperative} way of thinking. Using mixed periodic
strategies, we now want to analyze the game  within a
\textit{non-cooperative} perspective.

Let us analyze the mixed strategies that this game has.  There is no
mixed Nash equilibrium, but only the pure Nash strategy $(a_2,b_2)$.
The expected utilities of the two players for $p=p_N^*=0$, and
$q=q_N^*=0$ are both
\begin{equation}\label{mikjhxjghhg12}
{\mathcal{U}_A}_{p,q}(p_N^*=0,q_N^*=0)={\mathcal{U}_B}_{p,q}(p_N^*=0,q_N^*=0)=0.
\end{equation}
For identifying mixed periodic strategies, we maximize the expected
utility of player A with respect to $q$ and that  of player B with
respect to $p$. This results in  the pure strategies $p_p^*=1$ and
$q_p^*=1$. The expected utilities of both players are maximized for
this periodic strategy,
\begin{equation}\label{mikjhxjghhgfhfhf12}
{\mathcal{U}_A}_{p,q}(p_p^*=1,q_p^*=1)={\mathcal{U}_B}_{p,q}(p_p^*=1,q_p^*=1)=4
\end{equation}
Hence, in this case, the social optimum strategy is realized by
periodic strategies, although we have  used a non-cooperative method
in terms of a self-maximization procedure for each player. The fact
that the two outcomes, that is, non-cooperative and cooperative
ones, coincide depends on the particular details of this game. But
the important insight is that  the non-cooperative optimization
procedure underlying
 the periodic strategies
yields a higher payoff than  the Nash strategies.

We will further analyze this type of games, by exploiting another
example,  Game 4  in Table \ref{game4m}, which again is a collective
action game.
\begin{table}[h]
\centering
  \begin{tabular}{| l |l |l | }
    \hline
       &  $b_1$ & $b_2$ \\ \hline
  $a_1$ & 0,0 & 6,1   \\ \hline
    $a_2$ & 1,6 & 3,3\\
    \hline
 \end{tabular}
\caption{Mixed Strategies Game 4}\label{game4m}
\end{table}
The mixed Nash equilibrium is $p_N^*=\frac{3}{4},q_N^*=\frac{3}{4}$.
If we maximize A's  expected utility subject to $q$ we get
$\frac{\partial \mathcal{U}_A}{\partial q}=-2<0$ $\forall$ $q$.
Hence, the expected utility is maximized when $q=q_p^*=0$, since the
utility is monotonically decreasing with respect to $q$.
Correspondingly, maximizing  B's expected utility with respect to
$p$ we get $\frac{\partial \mathcal{U}_B}{\partial p}=-2<0$
$\forall$ $p$, and  B's  expected utility is maximized when
$p_p^*=0$. Therefore, the periodic strategies are $p_p^*=0,q_p^*=0$.
For  the periodic strategies,
${\mathcal{U}_B}_{p,q}(p_p^*=0,q_p^*=0)=3$ and
${\mathcal{U}_A}_{p,q}(p_p^*=0,q_p^*=0)=3$ while for the Nash
strategies we get ${\mathcal{U}_A}_{p,q}(q_N^*=3/4,p_N^*=3/4)=2.25$
and ${\mathcal{U}_B}_{p,q}(q_N^*=3/4,p_N^*=3/4)=2.25$.

Hence, the periodic  strategy, which is  pure in this game,  again
yields higher payoffs for both players than  the  Nash strategy.

In conclusion, this provides a possible solution to the question
\textit{why the players should play  socially optimal strategies},
instead of individually optimal ones,  \textit{through a strictly
non-cooperative scheme}.

\section{Two-player Simultaneous Strategic Form Games with a Continuum Set of Strategies}

In this section, we shall study the implications of the periodic
strategies algorithm for the case of two player games with
continuous strategies for the players. We shall consider  strategic
form, simultaneous, symmetric games with quadratic payoffs. Many
examples from the economics literature  belong to this class of
games, such as the Cournot and Bertrand duopoly, provision of public
good and search games \cite{moulinraygupta}. while  the periodic
strategies algorithm does not yield as interesting results as  in
the collective action games, we shall nevertheless present it in
order to explore possible applications of the algorithm.

We consider a game with two players $I=1,2$, for which a continuum
set of strategies is available for each player and where the payoffs
are  of the following form:
\begin{align}\label{generalform}
&u_1(x,y)=a_1x+a_2y+a_3xy+a_4x^2+a_5y^2 \\ \notag &
u_2(x,y)=b_1x+b_2y+b_3xy+b_4x^2+b_5y^2
\end{align}
where $x,y$ $\in$ $\mathbb{R}^{+}$. In both cases, the parameters
$a_5,b_5$ are assumed to be negative, to make  the payoffs concave
in the own strategy. For simplicity, we assume that  the game is
symmetric, that is $a_i=b_i$ for all $i$. For  the Nash equilibria,
the following equations must be solved simultaneously :
\begin{align}\label{nashmax}
&\frac{\partial{u_1}}{\partial x}=0,{\,}{\,}{\,}
\frac{\partial{u_2}}{\partial y}=0.
\end{align}
For the periodic strategies, we
  have to  solve simultaneously:
\begin{align}\label{permax}
\frac{\partial{u_1}}{\partial
y}=0,{\,}{\,}{\,}\frac{\partial{u_2}}{\partial x}=0.
\end{align}
In order  for the critical points $(x_N,y_N)$ of equations
(\ref{nashmax}) to be maxima,   the following two conditions have to
be satisfied:
\begin{align}\notag
&D_1=\frac{\partial^2u_1}{\partial x^2}\frac{\partial^2u_1}{\partial y^2}-\frac{\partial^2u_1}{\partial y\partial x}>0, & \frac{\partial^2u_1}{\partial x^2}\rvert_{x_N,y_N}<0\\
&D_2=\frac{\partial^2u_2}{\partial x^2}\frac{\partial^2u_2}{\partial
y^2}-\frac{\partial^2u_2}{\partial y\partial x}>0, &
\frac{\partial^2u_2}{\partial x^2}\rvert_{x_N,y_N}<0.
\end{align}
In the case of general quadratic games, the above conditions become:
\begin{align}
&D_1=-a_3^2+4a_4a_5>0,{\,}{\,}{\,}a_4<0 \\ \notag &
D_2=-b_3^2+4b_4b_5>0,{\,}{\,}{\,}b_4<0
\end{align}
The conditions for maxima  of the periodic strategies algorithm are
\begin{align}
&D_1=-a_3^2+4a_4a_5>0,{\,}{\,}{\,}a_5<0 \\ \notag &
D_2=-b_3^2+4b_4b_5>0,{\,}{\,}{\,}b_5<0.
\end{align}
For symmetric games, the conditions simplify:
\begin{align}
D_1=-a_3^2+4a_4a_5>0,{\,}{\,}{\,}a_4<0,{\,}{\,}{\,}a_5<0.
\end{align}
These last conditions are satisfied for all quadratic games that
satisfy   the convexity condition  imposed above. The Nash
equilibria and the ``periodic'' points are:
\begin{align}
&
x=-\frac{-2a_1b_5-a_3b_2}{-a_3b_3+4a_4b_5},{\,}{\,}{\,}y=-\frac{-2a_4b_2-a_1b_3}{-a_3b_3+4a_4b_5},{\,}{\,}{\,}\mathrm{Nash}
\\ \notag &
x=-\frac{-2a_5b_1+a_2b_3}{a_3b_3-4a_5b_4},{\,}{\,}{\,}y=-\frac{a_3b_1-2a_2b_4}{a_3b_3-4a_5b_4},{\,}{\,}{\,}\mathrm{Periodic}
\end{align}
We  shall now consider two examples of continuous, symmetric
quadratic games,  the ``Cournot Duopoly'' game and the ``Provision
of Public Good'' game.

\subsection{Cournot Duopoly}

The Cournot duopoly quadratic game has the  form
\begin{align}\label{cournotduopoly}
& u_1(x,y)=(P-A(x+y))x-(Bx-Mx^2),\\ \notag &
u_2(x,y)=(P-A(x+y))y-(By-My^2)
\end{align}
Thus, in  (\ref{generalform}) (with $a_i=b_i$), we put
\begin{equation}\label{pardefinition}
a_1=P-B,{\,}{\,}{\,}a_3=-A,{\,}{\,}{\,}a_4=-A+M,{\,}{\,}{\,}a_2=a_5=0.
\end{equation}
The Nash equilibrium is obtained by putting $\frac{\partial
u_1}{\partial x}= 0=\frac{\partial u_2}{\partial y}$, resulting in

\begin{equation}\label{criticalpoints}
x^*_N=y^*_N=\frac{P-B}{3A-2M}
\end{equation}
and the corresponding utilities are:
\begin{equation}\label{tiulicourn}
u_1(x^*_N,y^*_N)= u_2(x^*_N,y^*_N)=\frac{(P-B)^2(A-M)}{(3A-2M)^2}.
\end{equation}
In contrast, for the periodic equilibrium, we need $\frac{\partial
u_2}{\partial x}= 0=\frac{\partial u_1}{\partial y}$, resulting in
\begin{equation}\label{tiulicourn2}
x^*_p=y^*_p=0\text{ and }u_1(x^*_p,y^*_p)= u_2(x^*_p,y^*_p)=0.
\end{equation}
Thus, if $A>M$, the Nash equilibrium results in higher utilities for
the players. This game, however, is degenerate in the sense that
when computing the periodic optimum for A by putting
$0=\frac{\partial u_2}{\partial x}= -Ax$, this does not determine
the optimal value of the opponent's strategy $y$. In fact, for this
game, in (\ref{generalform}), the coefficient $a_5$ of the quadratic
term vanishes by \eqref{pardefinition}, and this causes the
degeneracy.

\subsection{Provision of Public Good Games}

The Provision of Public Good quadratic game has the  form
\begin{align}\label{cournotduopoly2}
& u_1(x,y)=A(x+y)-Cx-B(x+y)^2,\\ \notag &
u_2(x,y)=A(x+y)-Cy-B(x+y)^2.
\end{align}
The Nash equilibria have to satisfy
\begin{equation}\label{criticalpoints2}
x^*_N+y^*_N=\frac{A-C}{2B}
\end{equation}
with  utilities
\begin{equation}\label{tiulicourn2}
u_1(x^*_N,y)=\frac{(A-C)^2}{4B}+Cy, \quad
u_2(x,y^*_N)=\frac{(A-C)^2}{4B}+Cx.
\end{equation}
For the periodic equilibria, we get
\begin{equation}\label{pewer}
x^*_p+y^*_p=\frac{A}{2B}
\end{equation}
with  utilities
\begin{equation}\label{strategiesperiodic}
u_1(x^*_p,y)=\frac{A^2-2AC}{4B}+Cy , \quad
u_2(x,y^*_p)=\frac{A^2-2AC}{4B}+Cx.
\end{equation}
Thus, the periodic utilities are smaller than the Nash ones, since
\begin{equation}\label{combinedutilites}
u_1(x^*_p,y)=u_1(x^*_N,y)-\frac{C^2}{4B}, \quad
u_2(x,y^*_p)=u_2(x,y^*_N)-\frac{C^2}{4B},
\end{equation}
unless $C=0$. At the periodic equilibrium, the players invest too
much, as they ignore their own costs $C$. Again, however, this is a
special case, because of the symmetry of the game.

\section*{Concluding Remarks}

We have presented an intrinsic property of multiplayer, finite,
simultaneous, strategic form games, which we called periodicity of
strategies. We studied the periodicity concept in finite action
games and  proved that every finite action two player strategic form
game has at least one periodic action. Moreover, we proved that the
set of periodic strategies is set stable under the map
$\mathcal{Q}$. An action of some player $x_i$ is periodic if the
operator $\mathcal{Q}$ satisfies $\mathcal{Q}^{n_i}x_i=x_i$, where
$n_i$ is the periodicity number. When mixed strategies are taken
into account,  the  period length reduces to 1. While those
equilibria are different from the Nash ones, the resulting payoffs
for the players are the same, or in some cases, even higher than the
Nash payoffs.  Importantly, the periodic mixed strategy gives
outcomes for each player which do not depend on what the opponent
will play, in contrast to the Nash mixed strategy, where when the
opponent plays equilibrium, the own action does not matter. Thus, we
have found  a solution concept that in that regard is preferable to
the Nash equilibrium. More generally, we have compared our periodic
strategies to the rationalizable strategies of Bernheim
\cite{bernheim} and Pearce \cite{pearce}.

Moreover, the application of the algorithm to collective action
games gives another interesting result. We were able to demonstrate
that the social optimum strategy can be played by adopting a
non-cooperative thinking.

Also, the periodic equilibria can be identified by a straightforward algorithm. \\
                                                                                                                                                                                                                                                                                                    The next step would be the inclusion of mixed strategies in multiplayer games. In that case, the situation will become more complex because  the players could form coalitions. Periodicity then has to be reconsidered under this perspective. Also, one may wish to study Bayesian games with imperfect information from the perspective of our periodicity paradigm.

\end{document}